\documentclass[11pt]{article}
\usepackage[latin1]{inputenc}
\usepackage{graphicx}
\usepackage{epsfig}
\usepackage{color}
\usepackage{a4wide}
\usepackage{multirow}
\usepackage{lscape}
\usepackage{verbatim}
\usepackage{amsthm, amssymb}
\usepackage{amsmath}
\newtheorem{Lem}{Lemma}
\newtheorem{theorem}{Theorem}

\title{Approximate Distance Oracles with Improved Preprocessing Time}
\author{Christian Wulff-Nilsen
        \footnote{Department of Mathematics and Computer Science,
                  University of Southern Denmark,
                  \texttt{koolooz@diku.dk},
                  \texttt{http://www.imada.sdu.dk/$_{\widetilde{~}}$cwn/}.}}

\date{}
\begin{document}

\maketitle
\begin{abstract}
Given an undirected graph $G$ with $m$ edges, $n$ vertices, and non-negative edge weights, and given an integer $k\geq 1$,
we show that for some universal constant $c$, a $(2k-1)$-approximate distance oracle for $G$ of size $O(kn^{1 + 1/k})$ can be constructed in
$O(\sqrt km + kn^{1 + c/\sqrt k})$ time and can
answer queries in $O(k)$ time. We also give an oracle which is faster for smaller $k$. Our results break the quadratic preprocessing
time bound of Baswana and Kavitha for all $k\geq 6$ and improve the $O(kmn^{1/k})$ time bound of Thorup and Zwick except for very
sparse graphs and small $k$. When $m = \Omega(n^{1 + c/\sqrt k})$ and $k = O(1)$, our oracle is optimal w.r.t.\ both
stretch, size, preprocessing time, and query time, assuming a widely believed girth conjecture by Erd\H{o}s.
\end{abstract}
\newpage

\section{Introduction}\label{sec:Intro}
Computing shortest path distances in graphs is a fundamental algorithmic problem and has received a lot of
attention for several decades. Classical algorithms include that of Dijkstra which handles graphs with non-negative edge weights,
and the algorithm of Bellman-Ford which is slower but applies also when negative edge weights are present. A considerable drawback
of these algorithms is that they are too slow for many applications. For instance, a GPS system needs to report shortest path
distances extremely fast in very large road networks. If Dijkstra's algorithm is used, in the worst case the entire network would
need to be explored just to compute a single distance. Another problem is that the whole graph would have to be stored in memory. If
the graph is dense, it might not fit in the main memory and this would slow down computations considerably.

One way to speed up computations is to precompute distances between all pairs of vertices in a preprocessing step and store
them all in a look-up table. Distance queries can then be answered in constant time. The fastest known all-pairs shortest paths
algorithm has only marginally subcubic running time~\cite{Chan}. For sparser graphs, repeated applications of Dijkstra's algorithm
yield $O(mn + n^2\log n)$ time. Pettie~\cite{Pettie} gave a slightly improved bound of $O(mn + n^2\log\log n)$. Even for sparse graphs,
these algorithms are too slow for many applications. Another disadvantage of this scheme is the huge amount of memory required to
store the look-up table for all the distances.

\subsection{Approximate distances}
A way to deal with these issues is to settle for \emph{approximate} shortest path distances. For a directed graph $G$, a distance
from a vertex $u$ to a vertex $v$ along some path in $G$ is of \emph{stretch} $\delta\geq 1$ if the path is at most $\delta$ times
longer than a shortest path from $u$ to $v$ in $G$.

Zwick~\cite{Zwick} showed how to compute all-pairs stretch
$(1 + \epsilon)$-distances in directed graphs in $\tilde O(n^{2.376})$ time for an arbitrarily small constant $\epsilon > 0$.
In the seminal paper of Thorup and Zwick~\cite{ThorupZwick}, it was shown how to preprocess an undirected graph in close to linear
time to build a data structure of near-linear
size which reports small-stretch distances very fast. More precisely, for an undirected graph $G$ with non-negative
edge weights, $m$ edges, and $n$ vertices and for any integer $k\geq 1$, a data structure of size $O(kn^{1 + 1/k})$ can be built in
$O(kmn^{1/k})$ time which gives distances of stretch at most $2k-1$ in $O(k)$ time. Since this data structure has constant query time
for $k = O(1)$, it is referred to as an \emph{approximate distance oracle}. We emphasize that the result only holds for undirected graphs; as
shown by Thorup and Zwick, no compact distance oracles exist in general for directed graphs.

The Thorup-Zwick oracle is randomized. Roditty, Thorup, and Zwick~\cite{DetOracleSpanner} showed how to obtain a deterministic
algorithm with only a polylogarithmic increase in preprocessing time.
Mendel and Naor~\cite{MendelNaor} showed how to improve query time of the Thorup-Zwick oracle to $O(1)$ and space to
$O(n^{1 + 1/k})$ at the cost of a larger stretch of $O(k)$ and a longer preprocessing time of $O(n^{2 + 1/k}\log n)$.

For small $k$, the size/stretch tradeoff of the Thorup-Zwick oracle is essentially optimal due to a (widely believed and partially
proved) girth conjecture of Erd\H{o}s from $1963$~\cite{Erdos}. The only possible improvement is thus in the time for preprocessing.
It was shown by Baswana and Kavitha~\cite{APASP} that an oracle with the same size and stretch can be computed in $O(n^2)$ time which is
an improvement when the number of edges $m$ is $\omega(n^{2 - 1/k}/k)$ and it is the first essentially optimal oracle for
$m = \Theta(n^2)$.

\subsection{Contributions of this paper}
The main contribution of this paper is to break the quadratic time bound of Baswana and Kavitha~\cite{APASP} for $k\geq 6$ and
$m = o(n^2)$ and to break the Thorup-Zwick bound~\cite{ThorupZwick} when the graph is not too sparse and $k$ not too small. We show
that there exists a constant $c$ such that for any integer $k\geq 1$, a $(2k-1)$-approximate distance oracle of size
$O(kn^{1 + 1/k})$ and with query time $O(k)$ can be constructed in $O(\sqrt km + kn^{1 + c/\sqrt k})$ time. When
$m = \Omega(n^{1 + c/\sqrt k})$ and $k = O(1)$, preprocessing time is linear and our construction is thus optimal
in \emph{every} respect (stretch, size, preprocessing time, and query time), assuming the girth conjecture. The
oracle of Baswana and Kavitha only has linear preprocessing for $m = \Theta(n^2)$ and the $O(kmn^{1/k})$
preprocessing of Thorup and Zwick is super-linear for any constant $k$.

We also present an oracle which
is faster for smaller $k$. When $k\geq 3$, $k\bmod 3 = 0$, its preprocessing time is $O(km + kn^{3/2 + 2/k})$. We get similar bounds
for other values of $k\geq 3$: when $k\bmod 3 = 1$, preprocessing is $O(km + n^{3/2 + 1/(2k) + 3/(2(k-1))})$ and when $k\bmod 3 = 2$,
it is $O(km + n^{3/2 + 2/(k-2) - 1/(k(k-2))})$. In particular, we achieve subquadratic preprocessing for all $k\geq 6$.

Our algorithms are very simple to describe and analyze, given previous black boxes. As in earlier approaches, we make use of random
vertex sampling. We apply a result of Baswana and Kavitha~\cite{APASP} to sparsify our graph
w.r.t.\ this sampling while preserving \emph{exact} distances between pairs of vertices that are close in some sense. We construct a
Thorup-Zwick oracle for this sparser subgraph, allowing us to report $(2k-1)$-approximate distances between such vertex pairs.
For pairs that are farther apart, we
make use of spanners. For a $\delta\geq 1$, a \emph{$\delta$-spanner} of a graph $G$ is a subgraph that spans all vertices and
preserves distances between all vertex pairs up to a factor of $\delta$. To construct our oracle with small preprocessing time
for small $k$, we run the linear-time algorithm of Baswana and Sen~\cite{Spanner} to get a spanner of small stretch.
We apply Dijkstra to get exact distances in this spanner between all pairs of sampled vertices. These distances together with
distances from vertices to their nearest sampled vertex in the original graph allows us to report stretch $(2k-1)$-distances also
for vertices far apart. 

Our oracle with near-linear preprocessing time for larger $k$ does not run Dijkstra but instead constructs, on top of a spanner,
what we call a \emph{restricted} oracle. This oracle only allows us to query distances between sampled vertices
but is more compact than the Thorup-Zwick oracle. We pick the stretch of the spanner and the restricted oracle to be
$\Theta(\sqrt k)$ and we show that the oracle gives stretch $(2k-1)$-distances in the underlying graph.

We have summarized previous results on distance oracles as well as our new results in Table~\ref{tab:Oracles}.
\begin{table}
\caption{Performance of distance oracles in weighted undirected graphs.}\label{tab:Oracles}
\begin{center}
\begin{tabular}{|c|c|c|c|c|}
\hline
Stretch        & Query time & Space      & Preprocessing time      & Reference\\
\hline
$1$            & $O(1)$     & $O(n^2)$   & $O(mn + n^2\log\log n)$ & ~\cite{Pettie} \\
\hline
$1 + \epsilon$ & $O(1)$     & $O(n^2)$   & $\tilde O(n^{2.376})$   & ~\cite{Zwick} \\
\hline
\multirow{3}{*}{$2$} & \multirow{2}{*}{$O(1)$} & \multirow{2}{*}{$O(n^2)$}  & $O(n^{3/2}\sqrt m\log n)$   & ~\cite{CohenZwick} \\
                     &                         &                            & $O((m\sqrt n + n^2)\log n)$ & ~\cite{APASP}\\
\cline{2-5}
                  &                 $O(1)$  &              $O(n^{5/3})$  & --                          & ~\cite{Patrascu}\\
\hline
\multirow{2}{*}{$7/3$} & \multirow{2}{*}{$O(1)$} & \multirow{2}{*}{$O(n^2)$} & $O(n^{7/3}\log n)$          & ~\cite{CohenZwick}\\
                       &                         &                           & $O((m^{2/3}n + n^2)\log n)$ & \cite{APASP}\\
\hline
\multirow{4}{*}{$3$} & $O(1)$ & $O(n^2)$      & $O(n^2\log n)$ & ~\cite{CohenZwick}\\
                     & $O(1)$ & $O(n^{3/2})$  & $O(m\sqrt n)$  & ~\cite{ThorupZwick}\\
                     & $O(k)$ & $O(n^{3/2}))$ & $O(\min\{m\sqrt n,kn^{2 + \frac 1{2k}}\})$ & ~\cite{APASP} \\
                     & $O(\sqrt m)$ & $O(m + n)$ & -- & ~\cite{Patrascu}\\
\hline
\multirow{2}{*}{$2k-1$} & \multirow{2}{*}{$O(k)$} & \multirow{2}{*}{$O(kn^{1 + \frac 1 k})$} &
$O(kmn^{1 + \frac 1 k})$ & ~\cite{ThorupZwick}\\
&&& $O(\sqrt k m + kn^{1 + \frac c{\sqrt k}})$ & this paper\\
\hline
$2k-1$    & \multirow{2}{*}{$O(k)$} & \multirow{2}{*}{$O(kn^{1 + \frac 1 k})$} &
\multirow{2}{*}{$O(\min\{n^2,kmn^{1 + \frac 1 k}\})$} & \multirow{2}{*}{~\cite{APASP}}\\
($k\geq 3$)&&&&\\
\hline
$2k-1$    & \multirow{2}{*}{$O(k)$} & \multirow{2}{*}{$O(kn^{1 + \frac 1 k})$} & \multirow{2}{*}{$O(km + kn^{\frac 3 2 + \frac 2 k})$}
          & \multirow{2}{*}{this paper}\\
($k\geq 3$, $k\bmod 3 = 0$)&&&&\\
\hline
$2k-1$    & \multirow{2}{*}{$O(k)$} & \multirow{2}{*}{$O(kn^{1 + \frac 1 k})$} &
\multirow{2}{*}{$O(km + kn^{\frac 3 2 + \frac 1{2k} + \frac 3{2(k-1)}})$} & \multirow{2}{*}{this paper}\\
($k\geq 3$, $k\bmod 3 = 1$)&&&&\\
\hline
$2k-1$    & \multirow{2}{*}{$O(k)$} & \multirow{2}{*}{$O(kn^{1 + \frac 1 k})$} &
\multirow{2}{*}{$O(km + kn^{\frac 3 2 + \frac 2{k-2} - \frac 1{k(k-2)}})$} & \multirow{2}{*}{this paper}\\
($k\geq 3$, $k\bmod 3 = 2$)&&&&\\
\hline
\multirow{2}{*} {$O(k)$} & $O(1)$ & $O(n^{1 + 1/k})$ & $O(n^{2 + 1/k}\log n)$ & ~\cite{MendelNaor}\\
                       & $O(k)$ & $O(kn^{1 + 1/k})$ & $O(km + kn^{1 + 1/k + \epsilon})$ & this paper\\
\hline
\end{tabular}
\end{center}
\end{table}

\subsection{Related work}
A problem related to distance oracles is that of finding spanners. We have already mentioned the linear-time algorithm of
Baswana and Sen~\cite{Spanner} to find a spanner of stretch $2k-1$. There has also been interest in so-called
$(\alpha,\beta)$-spanners,
where $\alpha$ and $\beta$ are real numbers. Such a spanner $H$ of a graph $G$ ensures that for all vertices $u$ and $v$,
$d_H(u,v)\leq\alpha d_G(u,v) + \beta$. In other words, $H$ allows an \emph{additive} stretch in addition to a multiplicative
stretch. Thorup and Zwick~\cite{ThorupZwickSpanner} showed the existence of $(1 + \epsilon,\beta)$-spanners of size $O(n^{1 + 1/k})$
for any constant $k$, where $\beta = (c/\epsilon)^k$ for some constant $c$. A $(1,2)$-spanner of size $\tilde O(n^{3/2})$ was
presented by Dor, Halperin, and Zwick~\cite{DorSpanner}. The size was later improved slightly by Elkin and Peleg to
$O(n^{3/2})$~\cite{ElkinSpanner}. Baswana, Kavitha, Mehlhorn, and Pettie~\cite{SmallAdditiveSpanner} gave a spanner of size
$O(n^{4/3})$ which has additive stretch $6$ and no multiplicative stretch. This is currently the smallest known spanner with
constant additive stretch and no multiplicative stretch.

Allowing additive stretch in oracles has also been considered. For unweighted graphs, Baswana, Gaur, Sen, and
Upadhyay~\cite{SubquadraticOracle} showed how to get subquadratic construction time by allowing constant additive stretch
in addition to a multiplicative stretch of $k$. P\u{a}tra\c{s}cu and Roditty~\cite{Patrascu} showed that for unweighted
graphs there exists an oracle of size $O(n^{5/3})$ which has multiplicative stretch $2$ and additive stretch $1$. Furthermore, they
showed that for weighted graphs with $m = n^2/\alpha$ edges, there is an oracle of size $O(n^2/\alpha^{1/3})$ with multiplicative
stretch $2$ and no additive stretch. Preprocessing time was not considered in~\cite{Patrascu}.

The organization of the paper is as follows. In Section~\ref{sec:Prelim}, we give some basic definitions and notation as well as some
tools that will prove useful. In Section~\ref{sec:Oracle}, we give the stretch $(2k-1)$-oracle which is fast for small $k$. Our
near-linear time oracle is then presented in Section~\ref{sec:NearLinTimeOracle}. Finally,
we make some concluding remarks in Section~\ref{sec:ConclRem}.

\section{Definitions, Notation, and Toolbox}\label{sec:Prelim}
Let $G = (V,E)$ be an undirected graph with non-negative edge weights. For our problem, we may assume that all edges have strictly
positive weight since zero-weight edges can always be contracted. Also, we shall only consider connected graphs. For vertices
$u,v\in V$, we denote by $d_G(u,v)$ the shortest path
distance in $G$ between $u$ and $v$.

For a real value $\delta\geq 1$, a \emph{$\delta$-spanner} of $G$ is a subgraph
$H = (V,E_H)$ of $G$ spanning all its vertices such that for any distinct vertices $u$ and $v$, $d_H(u,v)\leq \delta d_G(u,v)$.

Let $S$ be a non-empty subset of $V$. For a vertex $u$,
let $p_S(u)$ be the vertex of $S$ nearest to $u$ w.r.t.\ $d_G$ (ties are resolved arbitrarily). We denote by $B_S(u)$ the set of
vertices $v$ with $d_G(u,v) < d_G(u,p_S(u))$. Let $E_S(v)$ denote the set of edges incident to $v$ with weight less than
$d_G(v,p_S(v))$. We define $E_S = \cup_{v\in V} E_S(v)$ and $G_S = (V,E_S)$. We need the following two simple results.
\begin{Lem}\label{Lem:NearestSample}
Given an undirected graph $G = (V,E)$ with $m$ edges and $n$ vertices and given a non-empty subset $S$ of $V$,
$p_S(u)$ and $d_G(u,p_S(u))$ can be computed in $O(m + n\log n)$ time over all vertices $u\in V$.
\end{Lem}
\begin{proof}
Connect a new vertex $s$ with a zero-weight edge to each vertex of $S$. Run Dijkstra
(implemented with Fibonacci heaps) from $s$ in this augmented graph. For each vertex $u\in V$, $p_S(u)$ is the unique ancestor of
$u$ belonging to $S$ in the shortest path tree found and the distance from $s$ to $p_S(u)$ in the tree equals $d_G(u,p_S(u))$.
\end{proof}
\begin{Lem}\label{Lem:GS}
Given an undirected graph $G = (V,E)$ with $m$ edges and $n$ vertices and given a non-empty subset $S$ of $V$,
$G_S$ can be computed in $O(m + n\log n)$ time.
\end{Lem}
\begin{proof}
Apply Lemma~\ref{Lem:NearestSample} to identify $p_S(u)$ and $d_G(u,p_S(u))$ for each $u\in V$. Then $E_S(u)$ can be found in time
proportional to the degree of $u$. Hence, $G_S$ can be found in $O(m)$ time in addition to the $O(m + n\log n)$ time from
Lemma~\ref{Lem:NearestSample}.
\end{proof}
The following result is due to Baswana and Kavitha (see Lemmas 2.2 and 2.3 in~\cite{APASP}).
\begin{Lem}\label{Lem:Sparse}
Let $G = (V,E)$ be an undirected $n$-vertex graph with positive edge weights and let $S\subseteq V$, $S\neq\emptyset$.
For any two vertices $u,v\in V$, if $v\in B_S(u)$ then $d_{G_S}(u,v) = d_G(u,v)$. If $S$ is obtained by picking each
vertex independently with probability $p$, then $E_S$ has expected size $O(n/p)$.
\end{Lem}

\section{A $(2k-1)$-Approximate Distance Oracle}\label{sec:Oracle}
In this section, we present a $(2k-1)$-approximate distance oracle with subquadratic preprocessing time for $k\geq 6$.
As a warm-up, we first present an $O(k)$ stretch oracle with near-linear preprocessing. It is a trivial combination of the
linear time spanner of Baswana and Sen~\cite{Spanner} and the Thorup-Zwick oracle~\cite{ThorupZwick}. This idea is probably
quite common knowledge and was noted by Sen~\cite{DistOracleSpannerSurvey}. However, the bounds obtained do not appear to be stated
explicitly in the literature so we include them here. Later in this section and in Section~\ref{sec:NearLinTimeOracle}, we
shall refine this idea in order to get optimal tradeoff between size and stretch.
\begin{theorem}\label{Thm:BigStretch}
Let $G$ be an undirected graph with $m$ edges and $n$ vertices and let $\epsilon > 0$ be a constant. For any integer
$k\geq 1$, an $O(k)$-approximate distance oracle for $G$ of size $O(kn^{1 + 1/k})$ can be constructed in
$O(km + kn^{1 + 1/k + \epsilon})$ time and can answer distance queries in $O(k)$ time.
\end{theorem}
\begin{proof}
We compute in $O(km)$ time a spanner $H$ of $G$ with stretch at most $\lceil 1/\epsilon\rceil = O(1)$ and with
$m_H = O(n^{1 + 1/\lceil 1/\epsilon\rceil}) = O(n^{1 + \epsilon})$ edges using the linear time
algorithm of Baswana and Sen~\cite{Spanner}.
We then construct a Thorup-Zwick oracle of stretch $2k-1$ on top of $H$. It has size
$O(kn^{1 + 1/k})$, query time $O(k)$, and preprocessing time $O(km_Hn^{1/k}) = O(kn^{1 + 1/k + \epsilon})$.
Since it has stretch $O(k)$ and $H$ has stretch $O(1)$, the theorem follows.
\end{proof}

\subsection{Preprocessing}\label{subsec:Preproc}
We now present our $(2k-1)$-approximate distance oracle and start with the preprocessing step.
Each vertex is sampled with probability $p = n^{-i/k}$, for some $0 < i\leq k$ to be specified; we allow $i$ to
be a non-integer. Let $S$ be the set of sampled vertices. We construct $G_S$ in $O(m + n\log n)$ time using Lemma~\ref{Lem:GS}.
The expected size of $S$ is $pn = n^{1 - i/k}$ and by Lemma~\ref{Lem:Sparse}, the expected size of $E_S$ is
$O(n/p) = O(n^{1 + i/k})$. We can rerun the sampling until $|S| = \Theta(n^{1 - i/k})$ and $|E_S| = O(n^{1 + i/k})$; by Markov's
inequality, only a constant expected number of reruns is needed for this.

We compute and store both $p_S(u)$ and $d_G(u,p_S(u))$ for all $u\in V$.
By Lemma~\ref{Lem:NearestSample}, this can be done in $O(m + n\log n)$ time.

Next, we build the distance oracle of Thorup and Zwick for the graph $G_S = (V,E_S)$. This takes
$O(k|E_S|n^{1/k}) = O(kn^{1 + (i+1)/k})$ expected time and requires $O(kn^{1 + 1/k})$ space. For some integer $k'$ (to be
specified), we apply the $O(k'm)$ time algorithm of Baswana and Sen~\cite{Spanner} to find a $(2k'-1)$-spanner $H = (V,E_H)$
of $G$ with $|E_H| = O(k'n^{1 + 1/k'})$ edges. For each pair of sampled
vertices $p,q\in S$, we compute and store $d_H(p,q)$. The latter is done by running Dijkstra in $H$ from each sampled vertex. Implementing Dijkstra with Fibonacci heaps, this takes
a total of $O(|S|(|E_H| + n\log n)) = O(k'n^{2 + 1/k' - i/k})$ time. The space required to store
all the $S\times S$ distances is $O(|S|^2) = O(n^{2 - 2i/k})$.

It follows from the above that total expected preprocessing time is $O(k'm + kn^{1 + (i+1)/k} + k'n^{2 + 1/k' - i/k})$ and the
amount of space needed for our oracle is $O(kn^{1 + 1/k} + n^{2 - 2i/k})$.

\subsection{Answering a distance query}
Now, let us consider how to answer a distance query for vertices $u,v\in V$, given the above preprocessing.
In constant time, we look up vertices $p_S(u)$ and $p_S(v)$ as well as distances $d_G(u,p_S(u))$ and $d_G(v,p_S(v))$.
We first query the distance oracle associated with $G_S$ and get a distance estimate $\tilde d_1(u,v)$ in $O(k)$ time. 
We then obtain the precomputed value $d_H(p_S(u),p_S(v))$ in constant time and obtain another
distance estimate $\tilde d_2(u,v) = d_G(u,p_S(u)) + d_H(p_S(u),p_S(v)) + d_G(v,p_S(v))$. The smallest of $\tilde d_1(u,v)$ and
$\tilde d_2(u,v)$ is then output as the answer to the query.

\subsection{Bounding stretch}\label{subsec:StretchBound}
Let $\tilde d_G(u,v) = \min\{\tilde d_1(u,v),\tilde d_2(u,v)\}$ denote the distance estimate that our query step outputs.
We show in the following that for a suitable choice of $k'$,
$d_G(u,v)\leq \tilde d_G(u,v)\leq (2k-1)d_G(u,v)$. The first inequality is clear since both $\tilde d_1(u,v)$ and
$\tilde d_2(u,v)$ are the weights of some $u$-$v$-paths in $G$; in particular, they are both at least as long as a shortest
$u$-$v$-path. In the following, we show $\tilde d_G(u,v)\leq (2k-1)d_G(u,v)$.

If $u\in B_S(v)$ or $v\in B_S(u)$ then by Lemma~\ref{Lem:Sparse}, $d_{G_S}(u,v) = d_G(u,v)$.
In this case, the oracle for $G_S$ outputs $\tilde d_1(u,v)\leq (2k - 1)d_G(u,v)$ so our query algorithm will output
$\tilde d_G(u,v)\leq\tilde d_1(u,v)\leq (2k-1)d_G(u,v)$, as desired.

Now assume that $u\notin B_S(v)$ and $v\notin B_S(u)$. Then $d_G(u,p_S(u)),d_G(v,p_S(v))\leq d_G(u,v)$ and hence
\begin{align*}
d_H(p_S(u),p_S(v)) & \leq (2k' - 1)d_G(p_S(u),p_S(v))\\
                   & \leq (2k' - 1)(d_G(p_S(u),u) + d_G(u,v) + d_G(v,p_S(v)))\\
                   & \leq (6k' - 3)d_G(u,v),
\end{align*}
so $\tilde d_2(u,v) = d_G(u,p_S(u)) + d_H(p_S(u),p_S(v)) + d_G(v,p_S(v)) \leq (6k' - 1)d_G(u,v)$.
Hence, if we choose $k' = \lfloor k/3\rfloor$, our query algorithm gives the desired stretch, i.e.,
$\tilde d_G(u,v)\leq\tilde d_2(u,v)\leq (2k-1)d_G(u,v)$. Here we need to assume that $k\geq 3$ since $k'$ needs to be at least $1$.

\subsection{Running time}
Now, let us bound running time. Our query step runs in $O(k)$ time as mentioned above. With our choice of $k'$,
preprocessing takes $O(km + kn^{1 + (i+1)/k} + n^{2 + 1/\lfloor k/3\rfloor - i/k})$ time. We consider three cases,
$k\bmod 3 = 0$, $k\bmod 3 = 1$, and $k\bmod 3 = 2$, and we will fix a value of $i$ in each of them that minimizes preprocessing
time. In the first case, $k' = k/3$
and we set $i = k/2 + 1$, in the second case $k' = (k-1)/3$ and we set $i = (k-1)/2 + 3k/(2(k-1))$, and in the
third case $k' = (k-2)/3$ and we set $i = (k-2)/2 + (2k-1)/(k-2)$. The time and space bounds we obtain from these choices
are stated in the following theorem.
\begin{theorem}\label{Thm:SmallStretch}
Let $G$ be an undirected graph with $m$ edges and $n$ vertices and let $k\geq 3$ be an integer. If $k\bmod 3 = 0$,
a $(2k-1)$-approximate distance oracle for $G$ of size $O(kn^{1 + 1/k})$ can be constructed in $O(km + kn^{3/2 + 2/k})$ time and can
answer queries in $O(k)$ time. If $k\bmod 3 = 1$ resp.\ $k\bmod 3 = 2$, the same size and query bounds hold
and construction time is $O(km + kn^{3/2 + 1/(2k) + 3/(2(k-1))})$ resp.\ $O(km + kn^{3/2 + 2/(k-2) - 1/(k(k-2))})$.
\end{theorem}
We see that Theorem~\ref{Thm:SmallStretch} breaks the quadratic time bound of Baswana and Kavitha~\cite{APASP}
when $k\geq 6$.

\section{Near-Linear Time Oracle}\label{sec:NearLinTimeOracle}
In this section, we give our $(2k-1)$-approximate distance
oracle that breaks the preprocessing time of the previous section for larger $k$. It achieves
a time bound arbitrarily close to linear when $k$ is sufficiently large.
We shall modify the oracle of Section~\ref{sec:Oracle} and the modification we make is in the distance computations in the spanner
$H$ between sampled vertices in $S$. Instead of Dijkstra, we build an approximate distance oracle on top of $H$ and then use it
to report distances. However, we cannot apply this idea directly since we aim for $O(kn^{1 + 1/k})$ space; if we built the
Thorup-Zwick oracle of that size on top of $H$, it would give stretch $2k-1$ which for our oracle would be multiplied
with the stretch of $H$. Instead, we will use the fact that the oracle only needs to report distances between vertices
of $S$, allowing us to improve the Thorup-Zwick space/stretch tradeoff.

\paragraph{Overview of the Thorup-Zwick oracle:}
First, let us very briefly go through the algorithm of Thorup and Zwick. For an integer $\kappa\geq 1$, to build a size
$O(\kappa n^{1 + 1/\kappa})$ oracle with stretch $2\kappa - 1$, sets $A_0,\ldots,A_\kappa$ are formed, with
$V = A_0\supseteq A_1\supseteq A_2\ldots\supseteq A_\kappa = \emptyset$. For $i = 1,\ldots,\kappa-1$, set $A_i$ is formed by
picking each element of $A_{i-1}$ independently with probability $n^{-1/\kappa}$. For $i = 0,\ldots,\kappa - 1$, distance
$d_G(A_i,v) = \min_{w\in A_i} d_G(w,v)$ is computed for each vertex $v$ and a vertex $p_i(v)$ of $A_i$ achieving this distance is kept.
Then so called \emph{bunches} $B(v)$ are formed around each vertex $v$ and are defined by:
\[
  B(v) = \cup_{i = 0}^{\kappa - 1}\{w\in A_i\setminus A_{i+1} | d_G(w,v) < d_G(A_{i+1},v)\}.
\]

The oracle answers a query for the approximate distance between vertices $u$ and $v$ by repeatedly checking whether suitably
chosen vertices belong to bunch $B(u)$ or bunch $B(v)$. More precisely, it starts by setting $w := u$ and initializes a counter
$i := 0$. Then as long as $w\notin B(v)$, it increments $i$, swaps $u$ and $v$ and updates $w := p_i(u)$. When $w\in B(v)$ (which
will happen at some point) $d_G(w,u) + d_G(w,v)$ is returned as approximate distance.

\paragraph{Improving space:}
The space requirement of the Thorup-Zwick oracle is dominated by the size of the bunches $B(v)$. For our application, we consider
a \emph{restricted} version of this oracle where we only care about queries between pairs of vertices from set $S$. It follows
from the query step of the Thorup-Zwick oracle that
we only need to store bunches $B(v)$ for $v\in S$ to answer such queries.
As shown by Thorup and Zwick, each bunch has size $O(\kappa n^{1/\kappa})$.
Hence, to build our restricted oracle, we only need $O(\kappa|S|n^{1/\kappa})$ space.\footnote{In fact, using the source-restricted
distance oracle in~\cite{DetOracleSpanner}, we can reduce space further to $O(\kappa|S|^{1 + 1/\kappa})$ by keeping $O(|S|)$ bunches each of size $O(\kappa|S|^{1/\kappa})$. This also gives a slightly improved
preprocessing time. Unfortunately, $|S|$ is too big for this result to give any significant improvement of our oracle.}

\paragraph{Our oracle:}
We get our near-linear time oracle by combining ideas from Section~\ref{sec:Oracle} with the restricted oracle above.
As before, set $S$ is formed by sampling each vertex with probability $p = n^{-i/k}$ and we construct the Thorup-Zwick oracle for
graph $G_S = (V,E_S)$. We build a $(2k'-1)$-spanner $H = (V,E_H)$ with $O(k'n^{1 + 1/k'})$ edges. Now, instead of applying Dijkstra
to find the exact distance $d_H(p,q)$ in $H$ between each pair of samples $p,q\in S$, we build a restricted $(2\kappa - 1)$-stretch
oracle for $H$ w.r.t.\ set $S$. It is built with the same procedure as that of Thorup and Zwick but only bunches $B(v)$ for
$v\in S$ are kept. Total preprocessing for the oracles for $G_S$ and $H$ is then
$O(k'm + kn^{1 + (i+1)/k} + \kappa k'n^{1 + 1/k' + 1/\kappa})$. We shall specify $i$, $k'$, and $\kappa$ below.

To answer a distance query for vertices $u$ and $v$, we do as in Section~\ref{sec:Oracle} but instead of using the exact spanner
distance $d_H(p_S(u),p_S(v))$ we use that obtained by the restricted oracle.

The analysis for bounding the stretch for such a query is almost identical to that in
Section~\ref{subsec:StretchBound}. We only need to consider the case $u\notin B_S(v)$ and $v\notin B_S(u)$ and we get a stretch
bounded by $2 + 3(2k' - 1)(2\kappa - 1)$ since we apply a $(2\kappa-1)$-approximate distance oracle instead of Dijkstra.

The size of our new oracle is bounded by the size of the oracle for $G_S$ and the restricted oracle for $H$ (note that we
do not need to store $H$ after the restricted oracle has been built). The first oracle requires
$O(kn^{1 + 1/k})$ space and as we saw above, the restricted oracle requires
$O(|S|\kappa n^{1/\kappa}) = O(\kappa n^{1 - i/k + 1/\kappa})$ space.
Since we only allow space $O(kn^{1 + 1/k})$, we need $\kappa\geq k/(i+1)$. To minimize stretch, we pick $\kappa$ as small as
possible while satisfying this inequality, i.e., we pick $\kappa = \lceil k/(i+1)\rceil$. Substituting this for $\kappa$ and
requiring that stretch is at most $2k-1$, we get the inequality
\[
  2 + 3(2k' - 1)\left(2\left\lceil\frac k{i+1}\right\rceil - 1\right)\leq 2k-1\Leftrightarrow
  k'\leq\frac{k + 3\left(\left\lceil\frac k{i+1}\right\rceil - 1\right)}{6\left\lceil\frac k{i+1}\right\rceil - 3}.
\]
To minimize running time, we pick $k'$ as large as possible while satisfying this inequality, i.e., we set
$k' = \lfloor\frac{k + 3(\lceil\frac k{i+1}\rceil - 1)}{6\lceil\frac k{i+1}\rceil - 3}\rfloor$ (for this to make sense, we need
$k'\geq 1$ but we shall choose $i$ and $k$ so that this is ensured).

We are aiming for a preprocessing time of $O(k'm + kn^{1 + c/\sqrt n})$ for some constant $c > 0$.
We shall pick $c$ such that an $i$ (possibly non-integer) can be found satisfying
\begin{align}
  \frac {2\sqrt k}{c} \leq \frac k{i+1}\leq\left\lceil\frac k{i+1}\right\rceil\leq\frac {c\sqrt k}{18}.\label{ineq}
\end{align}
To ensure this, it suffices to require that $2\sqrt k/c + 1 \leq c\sqrt k/18$.
Multiplying by $c$ on both sides, we get
\[
  \frac{\sqrt k}{18}c^2 - c - 2\sqrt k \geq 0\Leftrightarrow c \geq \frac 9{\sqrt k} + 9\sqrt{\frac 1 k + \frac 4 9}.
\]
Hence, since $k\geq 1$, picking $c = 9 + 3\sqrt{13}$ will allow us to pick an $i$ satisfying inequalities~(\ref{ineq}).
We claim that this choice of $i$ gives the desired preprocessing time bound. Assume that $c/\sqrt k\leq 1$ since we are
only interested in subquadratic bounds. Then
\[
  k' = \left\lfloor\frac{k + 3\left(\left\lceil\frac k{i+1}\right\rceil - 1\right)}{6\left\lceil\frac k{i+1}\right\rceil - 3}
       \right\rfloor
     \geq \frac {k + 3\left(\left\lceil\frac k{i+1}\right\rceil - 1\right)}{\frac{c\sqrt k}3 - 3} - 1
     >    \frac k{\frac{c\sqrt k}3} - 1
     =    \frac{3\sqrt k - c}{c}
     \geq \frac {2\sqrt k}{c}.
\]
Since also $\kappa = \lceil k/(i+1)\rceil \geq k/(i+1)\geq 2\sqrt k/c$ and $\kappa k' = O(k)$, total preprocessing time is
\[
  O(k'm + kn^{1 + (i+1)/k} + \kappa k'n^{1 + 1/k' + 1/\kappa}) =
  O(\sqrt k m + kn^{1 + c/(2\sqrt k)} + kn^{1 + c/\sqrt k}) =
  O(\sqrt k m + kn^{1 + c/\sqrt k}),
\]
as desired.
We can now state the main result of the paper.
\begin{theorem}\label{Thm:NearLinTimeOracle}
Let $G$ be an undirected graph with $m$ edges and $n$ vertices and let $k\geq 1$ be an integer. Then a
$(2k-1)$-approximate distance oracle for $G$ of size $O(kn^{1 + 1/k})$ can be constructed in
$O(\sqrt km + kn^{1 + c/\sqrt k})$ time, for some constant $c$, and can answer distance queries in $O(k)$ time.
\end{theorem}

The bound we gave on constant $c$ above was not very tight. In the following, let us bound this constant
for large $k$. Instead of picking $i$ such that inequalities~(\ref{ineq}) are satisfied, pick it such that
$\kappa = \lceil k/(i+1)\rceil \geq k/(i + 1) = \sqrt {k/6}$. Then
\[
k' =    \left\lfloor\frac{k + 3\left(\left\lceil\frac k{i+1}\right\rceil - 1\right)}
                         {6\left\lceil\frac k{i+1}\right\rceil - 3}\right\rfloor
   \geq \left\lfloor\frac{k + 3\sqrt\frac{k}{6} - 3}{6\left\lceil\sqrt\frac{k}{6}\right\rceil - 3}\right\rfloor
   \geq \frac{k + 3\sqrt\frac{k}{6} - 3}{6\sqrt\frac{k}{6} + 3} - 1
   =    \frac{k - 3\sqrt\frac{k}{6} - 6}{\sqrt{6k} + 3},
\]
which is at least $1$ when $k\geq 29$. Hence
\[
  \frac 1{\kappa} + \frac 1{k'}\leq \sqrt\frac{6}{k} + \frac{\sqrt{6k} + 3}{k - 3\sqrt\frac{k}{6} - 6}
\]
and it follows that we can pick $c$ arbitrarily close to $2\sqrt 6$ if $k$ is bounded from below by a sufficiently large
constant.

\section{Concluding Remarks}\label{sec:ConclRem}
For an undirected graph $G$ with $m$ edges, $n$ vertices, and non-negative edge weights, and for an integer $k\geq 1$, the main
result of this paper is a $(2k-1)$-approximate distance oracle of $G$ having size $O(kn^{1 + 1/k})$ and $O(k)$ query time which
can be constructed in $O(\sqrt k m + kn^{1 + c/\sqrt k})$ time for some constant $c$. We also gave an oracle
with faster preprocessing for smaller $k$. Together, these two results break the quadratic preprocessing
time of Baswana and Kavitha for $k\geq 6$ and improve the $O(kmn^{1/k})$ bound of Thorup and Zwick when $G$ is not too sparse and
$k$ not too small.

Assuming the girth conjecture, our oracle is optimal in every respect when $m = \Omega(n^{1 + c/\sqrt k})$ and
$k = O(1)$ since then preprocessing time is linear. Whether linear preprocessing time is achievable also for
sparser graphs remains an open problem. Our oracles break the quadratic preprocessing bound when $k \geq 6$.
What is possible for smaller $k$? Finally, we believe that by using existing techniques, it should be possible to
derandomize our oracles at only a small increase in preprocessing time.

\end{document}